\documentclass[10pt, conference, compsocconf]{IEEEtran}

\usepackage{algorithmic}
\usepackage{algorithm}
\usepackage{amssymb,amsmath,amsthm}
\usepackage{graphicx}

\newtheorem{theorem}{Theorem}

\begin{document}
%
\title{Decentralized Probabilistic Auto-Scaling for Heterogeneous Systems}


\author{\IEEEauthorblockN{Bogdan Alexandru Caprarescu, Dana Petcu}
\IEEEauthorblockA{Research Institute e-Austria and West University of Timi\c{s}oara\\
Timi\c{s}oara 300223, Rom\^{a}nia\\
\{bcaprarescu, petcu\}@info.uvt.ro}
}

\maketitle

\begin{abstract}
The DEPAS (Decentralized Probabilistic Auto-Scaling) algorithm assumes an overlay network of computing nodes where each node probabilistically decides to shut down, allocate one or more other nodes or do nothing. DEPAS was formulated, tested, and theoretically analyzed for the simplified case of homogenous systems. In this paper, we extend DEPAS to heterogeneous systems.
\end{abstract}

\begin{IEEEkeywords}
auto-scaling; decentralized computing; randomized algorithms; cloud computing
\end{IEEEkeywords}

\section{Introduction}

As the main characteristic of cloud computing, the on-demand provisioning of hardware resources creates the premises for a theoretically infinite scalability of services deployed in the cloud \cite{armbrust2010}. As a result, auto-scaling has become an interesting topic in both academia and industry. On one hand, the researchers focus on complex solutions that optimize the resource consumption and the QoS of the services. On the other hand, some cloud providers, such as Amazon EC2 and Rightscale, offer policy-based auto-scaling solutions that can be easily configured by their customers. 

A common characteristic of the large majority of research and industrial solutions for auto-scaling consists in their centralization. In this way, the infinite scalability that is theoretically possible is jeopardized by the central component running the auto-scaling algorithm. This is not only a theoretical problem and the limitations of centralized management were experienced with industrial systems such as VMware \cite{meng2010}. Moreover, a central manager acts as a single point of failure, too.

Our aim is to provide an auto-scaling algorithm that is both scalable and fault tolerant. To achieve this goal, we take inspiration from P2P systems which proved to be highly scalable and robust \cite{keong2005}. Thus, in \cite{calcavecchia2012}, we described DEPAS, a decentralized probabilistic auto-scaling algorithm. DEPAS assumes an overlay network of nodes. Each node is a virtual machine that runs the same software comprising the functional service and a few non-functional components: overlay manager, load balancer, and auto-scaler. Each non-functional component runs a decentralized algorithm. The auto-scaler runs DEPAS. 

The parameters of DEPAS are the desired load, $L_0$, and the load variation, $\delta$. The goal of DEPAS is to maintain the average load of the system in the interval $(L_0 - \delta, L_0 + \delta)$. To do that, each node estimates the average load of the system and if it is not within that interval, then the node probabilistically executes the appropriate scaling action (i.e., remove itself if the load is lower than $L_0 - \delta$ or allocate additional nodes if the load is higher than $L_0 + \delta$). The main problem is how to compute the node-level scaling probability.

DEPAS was originally formulated and tested for the simplified case of homogenous systems in which all nodes have the same capacity \cite{calcavecchia2012}. The link between $\delta$, number of nodes, and the probability of allocating the right number of nodes was theoretically analyzed and we provided algorithms for finding either the minimum $\delta$ (for a given number of nodes) or the minimum number of nodes (for a given $\delta$) so that a minimum correctness probability is guaranteed \cite{caprarescu2012}.

Some tests with heterogeneous nodes were also performed in \cite{calcavecchia2012} under the assumption that each node randomly choses the capacity of the node to be added with the same distribution as the capacity distribution of the existing nodes. For example, in a system with $70\%$ nodes with a low capacity and $30\%$ nodes with a high capacity a newly allocated node would have $70\%$ chances to be a low capacity one and $30\%$ chances to be a high capacity node. However, this is a highly constrained scenario for which it is difficult to imagine a practical applicability. This is because in practice the type of a new virtual machine is chosen so that to optimize a certain client-defined criterion (for example, the VM type with the minimum cost per capacity unit). The selection of the optimal node type may be related to the criterion to select it, certain aspects of the cloud infrastructure, or the bid for resources from other customers, but it has nothing to do with the capcity distribution of the existing nodes. 

Therefore, in this paper, we extend DEPAS to heterogeneous systems without imposing any constraint on the capacity distribution of the new nodes. Thus, we provide a formula for computing the addition probability that works no matter the capacity of the new nodes. The correctness of the formulas for computing the addition and removal probabilities is proven both theoretically and experimentally.

The remaining of this paper is organized as follows. The DEPAS algorithm for heterogeneous systems is described in Section \ref{sec:depashe} and experimentally verified in Section \ref{sec:results}. Related work is discussed in Section \ref{sec:rw}, while Section \ref{sec:concl} concludes the paper.

\section{DEPAS for heterogeneous Systems}
\label{sec:depashe}

We assume a system composed of $n$ nodes with capacities $C_i, i = 1..n$. The capacity of a node is the maximum number of requests per second that can be processed by the service deployed on that node. Let $C$ be the total capacity of the system (the sum of the capacities of all nodes). A complete list of notations is given in Table \ref{table:notations}. The load of a node, noted with $L_i$, is computed at a given moment in time as a ratio between the average number of requests per second that were either processed or rejected by that node over a certain timeframe and the capacity of the node. Depending on the load balancing algorithm, a node may reject a request in certain cases (e.g., when the system is overloaded and the maximum response time of the request can not be met). Then, the average load of the system, $L$, is computed as a ratio between the workload of the system and the capacity of the system as expressed by equation (\ref{eq:av-load}). Note that in the case when the workload received by the system overcomes its capacity, the average load is supra-unitary.

\begin{equation}
\label{eq:av-load}
L = \frac{\sum_{i=1}^{n} L_i C_i}{\sum_{i=1}^{n} C_i} = \frac{\sum_{i=1}^{n} L_i C_i}{C}
\end{equation}

\begin{table}[b!]
	 \caption{DEPAS notations}
\begin{center}
\begin{tabular}{ | p{0.7cm} | p{6.3cm} | }
	\hline
	$T$ & The duration of a DEPAS cycle (in seconds)\\
	$n$ & Number of nodes of the system \\
	$C$ & Total capacity of the system \\
	$C_i$ & Capacity of node $i$ \\
	$C^{add}_{i}$ & Capacity of the nodes to be added by node $i$ in the current cycle \\
	$C^{rem}_{opt}$ & Optimal capacity to be removed \\
	$C^{add}_{opt}$ & Optimal capacity to be added \\
	$C^{rem}_{exp}$ & Expected capacity to be removed \\
	$C^{add}_{exp}$ & Expected capacity to be added \\
	$L_0$ & Desired load threshold (percent with respect to the capacity) \\
	$\delta$ & Defines the allowed load variation \\
	$L_i$ & Load of node $i$ (percent with respect to node capacity) \\
	$L$ & Average load of the system (percent with respect to the capacity) \\
	$L^{\star}_{i}$ & An estimation of the average load of the system done by node $i$ \\
	$pi_i$ & Probability indicator computed by node $i$ \\
	$p_i$ & Node-level probability computed by node $i$\\
	\hline
\end{tabular}
\end{center}
   
    \label{table:notations}
\end{table}

The DEPAS algorithm for heterogeneous systems is shown in algorithm \ref{alg:pas}. It is periodically run with period $T$ by each node and begins by retrieving an estimation of the average load of the system. Note that the average load is not computed at this time, but just retrieved from the component running the average protocol. If the load is less than or equal to $L_0 - \delta$, then the node computes a removal probability indicator using formula (\ref{eq:removal-prob-indicator}) and, because the indicator is sub-unitary in this case, the node uses it as the probability to remove itself. Otherwise, if the load is higher than or equal to $L_0 + \delta$, then the node obtains its own capacity and the capacity of the node type that is the most convenient to be allocated at this time. Then, the probability indicator is computed using formula (\ref{eq:addition-prob-indicator}). In this situation, the indicator can be supra-unitary where its integer part represents the number of nodes to be added for sure, while its fractional part is used as the probability to add another node. Note that the \textit{random()} function generates a uniformly distributed random decimal number between 0 and 1.

\begin{algorithm}[t!]
\caption{DEPAS for heterogeneous Systems}
\label{alg:pas}
\begin{algorithmic}
\WHILE{$true$}
\STATE  $wait(T)$
\STATE $L^{\star}_{i} \gets estimateAverageSystemLoad()$
\IF {$L^{\star}_{i} \le L_0 - \delta$}	
       \STATE $pi_i \gets computeRemovalProbInd(L^{\star}_{i}, L_0)$
       \STATE $p_i \gets pi_i$
	 \IF {$p_i < random()$}
        	\STATE $removeSelf()$
	\ENDIF
\ELSE
	\IF {$L^{\star} \ge L_0 + \delta$}
		\STATE $C_{i} = getSelfCapacity()$
		\STATE $C^{add}_{i} = computeNewNodesCapacity()$
      		\STATE $pi_i \gets computeAdditionProbInd(L^{\star}_{i}, L_0, C_{i}, C^{add}_{i})$
		\STATE $m \gets \lfloor pi_i \rfloor$
		\STATE $p_i \gets \left\{pi_i\right\}$
		\IF {$p_i < random()$}
        		\STATE $m \gets m + 1$
		\ENDIF
		\STATE $addNodes(m, C^{add}_{i})$
	\ENDIF
\ENDIF
\ENDWHILE
\end{algorithmic}
\end{algorithm}

\begin{equation}
\label{eq:removal-prob-indicator}
pi_i =  \frac{L_0 - L^{\star}_{i}}{L_0}
\end{equation}

\begin{equation}
\label{eq:addition-prob-indicator}
pi_i =  \frac{L^{\star}_{i} - L_0}{L_0} \frac{C_{i}}{C^{add}_{i}}
\end{equation}

The remaining of this section proves that the formulas (\ref{eq:removal-prob-indicator}) and (\ref{eq:addition-prob-indicator}) are correct. They are correct if the expected capacity to be removed or added is equal to the optimal capacity to be removed or added, respectively. The optimal capacity is defined as the amount of capacity that needs to be subtracted from or added to the system so that the new average load is equal to the desired load. Theorems \ref{th:optimal-removal-capacity} and \ref{th:optimal-addition-capacity} provides formulas for computing the optimal capacity to be removed and the optimal capacity to be added, respectively.

\begin{theorem}
\label{th:optimal-removal-capacity}
Let $L_0 \in (0,1)$ be the desired load. Consider a system with total capacity $C$ and average load $L < L_0$. Then, the optimal capacity to be removed from the system is computed as follows: 

$$C^{rem}_{opt} =  \frac{L_0 - L}{L_0} C$$
\end{theorem}

\begin{proof}
As the system has the same workload before and after removing capacity we have $L C = L_0 (C - C^{rem}_{opt})$, from where it results the formula given by the theorem.
\end{proof}

\begin{theorem}
\label{th:optimal-addition-capacity}
Let $L_0 \in (0,1)$ be the desired load. Consider a system with total capacity $C$ and average load $L > L_0$. Then, the optimal capacity to be added to the system is computed as follows: 

$$C^{add}_{opt} =  \frac{L - L_0}{L_0} C$$
\end{theorem}

\begin{proof}
Analogues with the proof of Theorem \ref{th:optimal-removal-capacity}.
\end{proof}

The expected capacity to be removed/added is computed for a cycle of DEPAS. A cycle has a duration of $T$ seconds in which each node runs DEPAS exactly once. In order to be able to compute the expected capacity we assume that each node precisely estimates the average load of the system, which means that $L^{\star}_{i} = L \forall i = 1..n$. Under these considerations, theorems \ref{th:removal-correctness} and \ref{th:addition-correctness} prove the correctness of the formulas for computing the removal and addition probabilities. 

\begin{theorem}
\label{th:removal-correctness}
Let $L_0 \in (0,1)$ be the desired load. Consider a system with $n$ nodes, total capacity $C$ and average load $L < L_0$. The expected capacity to be removed in a cycle of DEPAS is equal to the optimal capacity to be removed.
\end{theorem}

\begin{proof}
$$C^{rem}_{exp} = \sum_{i=1}^{n} \frac{L_0-L}{L_0} C_i = \frac{L_0-L}{L_0} C = C^{rem}_{opt}$$
\end{proof}

\begin{theorem}
\label{th:addition-correctness}
Let $L_0 \in (0,1)$ be the desired load. Consider a system with $n$ nodes, total capacity $C$ and average load $L > L_0$. The expected capacity to be added in a cycle of DEPAS is equal to the optimal capacity to be added.
\end{theorem}

\begin{proof}
$$C^{add}_{exp} = \sum_{i=1}^{n} \frac{L-L_0}{L_0} \frac{C_i}{C^{add}{i}} C^{add}{i} = \frac{L-L_0}{L_0} C = C^{add}_{opt}$$
\end{proof}

In a previous paper \cite{calcavecchia2012}, the addition probability indicator was computed with formula (\ref{eq:old-addition-prob-indicator}). We are interested in which conditions this formula leads to a correct allocation in a heterogeneous systems. A correct allocation happens when $C^{add}_{exp} = C^{add}_{opt}$. By replacing the expected and optimum capacities with their formulas and making the simplifications it results equation (\ref{eq:old-condition}).

\begin{equation}
\label{eq:old-addition-prob-indicator}
pi_i =  \frac{L-L_0}{L_0}
\end{equation}

\begin{equation}
\label{eq:old-condition}
\sum_{i=1}^{n} C^{add}_{i} = C
\end{equation}

From equation (\ref{eq:old-condition}) it turns out that formula (\ref{eq:old-addition-prob-indicator}) can be used to compute the addition probability in a heterogeneous system only if the sum of the potential capacities to be added by each node is equal to the total capacity of the system. This is obviously a very particular situation without any practical motivation.

Therefore, in this section, we provided a general formula -- expressed by equation (\ref{eq:addition-prob-indicator}) -- for computing the addition probability of the DEPAS algorithm in a heterogeneous system.

\section{Experimental Results}
\label{sec:results}

In the previous section we proved that the formulas for computing the removal and addition probabilities are correct providing that each node knows the average load of the system. In this section we relax this requirement and experimentally show that those formulas lead to good allocations even if the average load of the system is approximated at each node with the average load of the node and its neighbors.

For running the experiments we adapted the simulator that was used in a previous paper \cite{calcavecchia2012}. This simulator is based on the Protopeer library \cite{protopeer} and is available for download \cite{depashe}. The simulator is described in Subsection \ref{subsec:results:exp-settings}, while the results of the experiment are shown in Subsection \ref{subsec:results:exp-res}.

\subsection{Settings}
\label{subsec:results:exp-settings}

In the simulator we designed three types of peers: client, entry point, and worker. In the experiment that is described in this section we use only one client and one entry point. The client issues requests according to an exponential distribution whose mean value follows a given workload track (see Subsection \ref{subsec:results:exp-res}). The requests arrive at the entry point which knows a percent of all workers but not less than a given minimum. The workers known by the entry point are randomly selected and periodically renewed. The entry pointed dispatches each request to one worker according to a capacity-weighted random load balancing strategy. The values used in the experiment for the parameters that are discussed in this subsection are listed in Table \ref{table:sim-parameters}.

\begin{table}
	 \caption{Parameters of the simulator}
\begin{center}
\begin{tabular}{ | p{6.3cm} | p{0.7cm} | }
	\hline
	Min no of entry point neighbors & 50\\
	Percent of entry point neighbors & 2\%\\
	Entry point neighbor reshuffle period & 120s\\
	Overlay degree & 50\\
	Overlay management cycle & 0.5s\\
	Max queue size & 3\\
	Max no of hops & 10\\
	Mean execution time & 1\\
	Load monitoring period & 60s\\
	$T$ (DEPAS cycle duration) & 60s\\
	$L_0$ & 0.7\\
	$\delta$ & 0.1\\
	\hline
\end{tabular}
\end{center}
   
    \label{table:sim-parameters}
\end{table}

The nodes are organized into an unstructured overlay network where each node runs an overlay management algorithm to maintain a list of neighbors. We adapted the gossip-based overlay management algorithm developed by Jelasity et al. \cite{jelasity2005} to obtain two characteristics needed by our system: quick removal of links to dead nodes (needed by both load balancing and average load estimator) and low-deviation in-degree (needed for proper load balancing). The resulted overlay management algorithm will be described in a dedicated paper. The values of the main parameters related to overlay management (overlay degree and cycle duration) are given in Table \ref{table:sim-parameters}.

The first-level load balancing performed by the entry point is accompanied by a decentralized load balancing that is performed by the workers. The decentralized load balancing algorithm is taken from \cite{adam2006} and adapted for heterogeneous systems. The idea is that, when a worker receives a request (which may come from either an entry point or other worker), an admission function is called to decide whether the request is scheduled on the current node, routed to a neighbor, or rejected. More concretely, each node uses an internal queue to store the requests that are pending for execution and prioritize them depending on the time they were issued by the client (i.e., older requests get priority over newer requests). If the length of the queue divided by the capacity of the node is higher than a threshold then the request is added to the queue. Otherwise, if the number of nodes already visited by the request without being scheduled on neither of them (called number of hops) is less than a threshold, then the request is routed to a neighbor of the current node in the overlay network. The neighbor is selected using the same capacity-weighted random strategy that is used by the entry point. Finally, if the maximum number of hops is reached, then the request is rejected. 

The request execution time is exponentially distributed with a constant mean. The capacity of a node is expressed in the simulator using an integer number $c$ which means that the node can execute $c$ requests in parallel. The mean execution time is always the same no matter the capacity of the node.

The average load of a node is computed per second by counting the requests that were either processed or rejected by the current node over a timeframe called load monitoring period. The average load of the system is approximated with the average load of the current node and its neighbors. The load monitoring period and the DEPAS cycle duration (i.e., the time between two consecutive executions of DEPAS on the same node) are important parameters because they affect the reactivity and accuracy of DEPAS. On one hand, small values of these parameters make DEPAS to quickly react to workload changes but also very sensible to oscillations (i.e., additions and removals of nodes mixed in a row). Of course, the oscillations should be avoided because they waste the money of the customer. On the other hand, higher values of the load monitoring period and the DEPAS cycle make the system more stable at the cost of delaying its reactivity and thus rejecting more requests when confronted with workload bursts. 

Finally, the values of the desired load and load variation threshold are also given in Table \ref{table:sim-parameters}.

\subsection{Results}
\label{subsec:results:exp-res}

Provided that a list of fresh neighbor information is available at each node and that the system is properly load balanced, the DEPAS algorithm works properly and is very scalable and fault tolerant. It is very scalable because (i) it is very simple and is run once every $T$ seconds and (ii) it does not require any message exchange between neighbors. It is robust because it does not care about node crashes as the update of neighbors information is the task of the overlay management algorithm. Therefore, a complete solution based on DEPAS is scalable and robust as long as it employs scalable and robust overlay management and load balancing algorithms. But the scalability and robustness of these algorithms were already experimentally proved in \cite{calcavecchia2012}. This is why, in this paper, we experimentally check only the correctness of the DEPAS algorithm for heterogeneous systems. More concretely, for a dynamic workload scenario we verify that the allocated capacity is close to the optimum capacity.

The overlay management, load balancing, and auto-scaling algorithms are non-deterministic. Therefore, meaningful results can be obtained only by averaging the results of several identical, but independently performed experiments. The results presented in this paper are the average results of 32 identical experiments. All experiments were sequentially executed on one Amazon High-CPU Medium Instance (1.7 GB of memory, 2 virtual cores with 2.5 EC2 Compute Units each) running Amazon Linux 64-bit.

Figure \ref{fig:workload-track} shows the mean workload track used in our experiment. The actual workload is exponentially distributed with the dynamic mean given by the workload track. The simulated experiment lasts 2600 seconds. The mean workload is 70 requests/s at the beginning of the experiment and increases in a few steps to reach 2200 requests/s at its peak. Then, to also test the scale out branch of DEPAS, the mean workload is decreased in steps. 

We use two types of nodes: low capacity nodes with a capacity of 1 request/s and high capacity nodes with a capacity of 5 requests/s. The system is initialized with 100 low capacity nodes, which have the perfect capacity for handling the initial workload of 70 requests/s. The existing nodes can add high-capacity nodes in the first 1099 seconds and low capacity nodes after second 1100 inclusive (which marks the last workload burst as shown in Figure \ref{fig:workload-track}).

Figure \ref{fig:allocated-capacity} shows the capacity allocated by DEPAS versus the optimum capacity computed for the desired load ($L_0$), min load ($L_0 - \delta$), and max load ($L_0 + \delta$). It can be seen that, after a period of adaptation, DEPAS allocates a capacity that is between the optimum capacity at max load and optimum capacity and min load. The delay in adaptation is caused by the duration of the DEPAS cycle (60 seconds) and by the fact that the load is averaged over the last 60 seconds. But, as explained in Subsection \ref{subsec:results:exp-settings}, the good side of these setting consists in the fact that, despite its randomized and decentralized nature, DEPAS is very stable and did not cause any capacity oscillation.
\begin{figure}
        \includegraphics[scale=0.8]{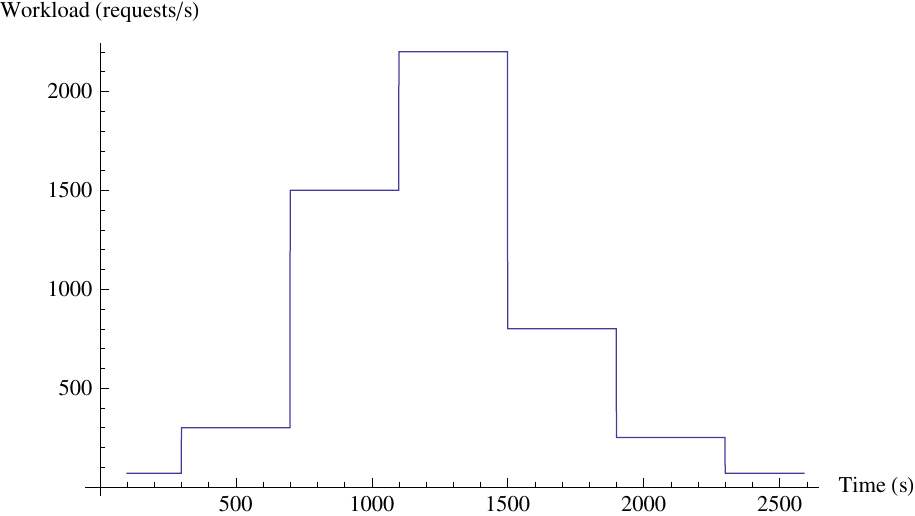}
        \caption{Mean workload track}
        \label{fig:workload-track}
\end{figure}

\begin{figure}
        \includegraphics[scale=0.8]{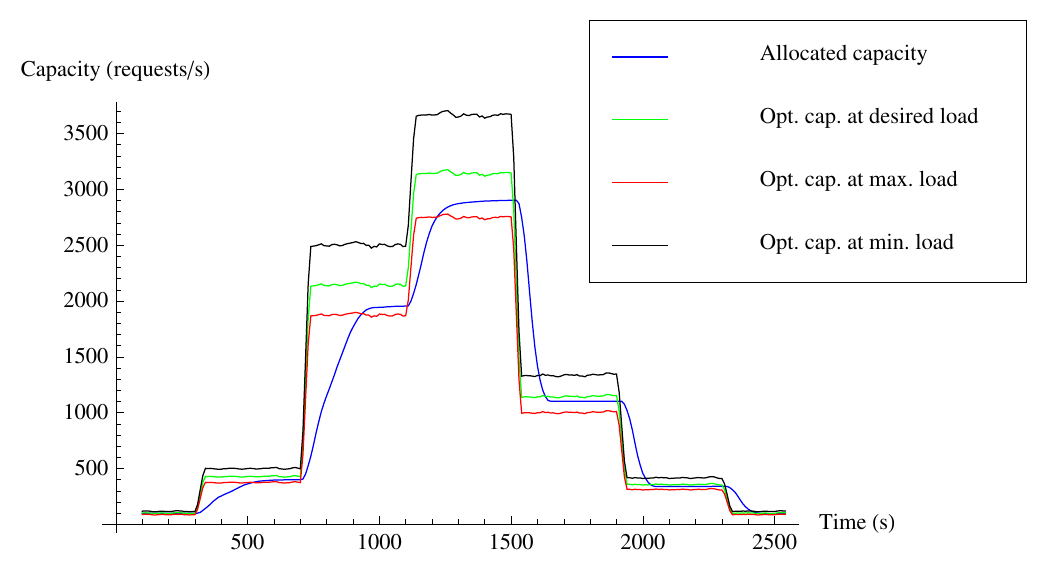}
        \caption{Allocated vs. optimal capacity}
        \label{fig:allocated-capacity}
\end{figure}

\begin{figure}
        \includegraphics[scale=0.8]{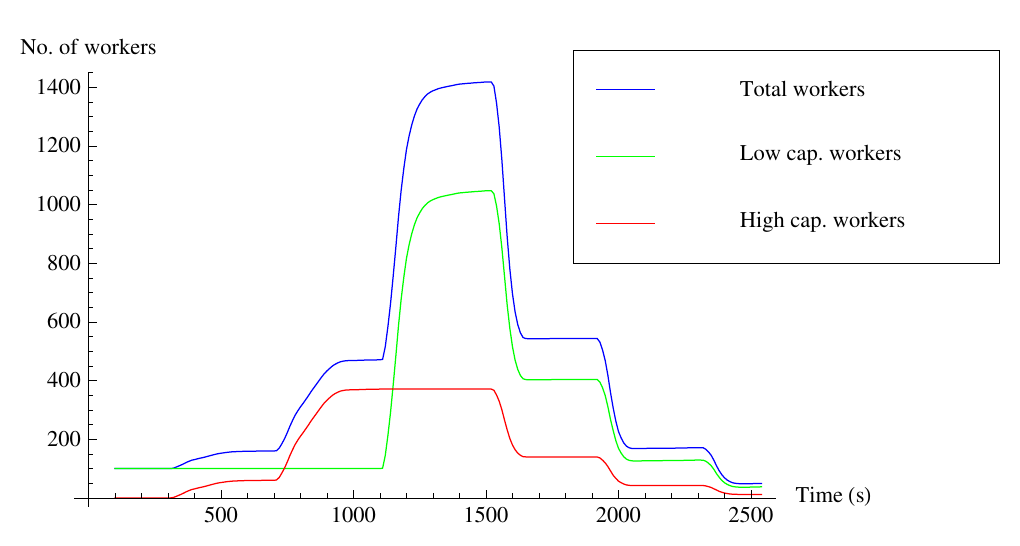}
        \caption{Number of nodes}
        \label{fig:nb-of-nodes}
\end{figure}

Figure \ref{fig:nb-of-nodes} shows the variation of the total number of nodes, number of low capacity nodes, and number of high capacity nodes. We can see that the system, initially composed of 100 low capacity nodes, adapts to the increasing workload by adding high capacity nodes before the second 1100 and low capacity nodes afterwards. Then, during the scale out, the system removes both low capacity and high capacity nodes. As expected, the percents of both types of nodes out of the total number of nodes seems to remain constant while scaling out. Figure \ref{fig:nb-of-nodes} confirms that DEPAS does not make any over-provisioning (i.e., allocation of more nodes that needed) or over-de-provisioning (i.e., removal of more nodes than needed).

In conclusion, the experimental results show that, facing a variable workload in a heterogeneous system, DEPAS allocates the right capacity even if each node works with a local approximation of the average load.

\section{Related Work}
\label{sec:rw}

In this paper we discuss a few other decentralized approaches to auto-scaling. A detailed state of the art on autonomic resource provisioning for cloud computing can be found in \cite{calcavecchia2012}.

A decentralized economic-inspired solution to the auto-scaling of component-based systems was proposed by Bonvin et al. \cite{bonvin2011}. They use a multi-agent approach in which each server is managed by a server agent. This agent makes decisions related to the migration/replication/removal of the components deployed on that server. The problem is that each agent stores a complete mapping (maintained through gossiping) of components and servers. In other words, each agent has a complete view of the system. Because of that, although the approach is decentralized in the sense that there is no central manager, it is not scalable with respect to the number of components and servers.

Another decentralized auto-scaling approach was proposed by Wuhib et al. \cite{wuhib2010}. They aim to develop a Platform as a Service for hosting sites in the cloud. In their approach, each virtual machine is managed by a VM manager which is connected through a custom overlay network to other VM managers that store instances of the same sites. The utility of a site instance is computed as the ratio between the allocated CPU capacity and the CPU demand. The utility of the system is the minimum utility of all instances of all sites. A decentralized heuristic algorithm is used to maximize the utility of the system while minimizing the cost of adaptation. The resulting system is scalable with respect to the number of virtual machines and the number of sites, but it is not scalable with respect to the number of instances of a site.

Montresor and Zandonati proposed a decentralized algorithm for selecting a slice of a P2P system \cite{montresor2008}. This slice may contain nodes with given characteristics that are needed for running a certain distributed application. Their approach shares the same idea with DEPAS: each node probabilistically decides tojoin the slice or depart from the slice. But their approach is more complex and less scalable than DEPAS because they use an epidemic broadcast algorithm to inform all nodes about the slice to be created and a peer counting algorithm that provides each node with an estimation of the slice size. The key to the high scalability of DEPAS is that a node does not need to know either the total number of nodes or the total capacity of the system.  

\section{Conclusion}
\label{sec:concl}

The Decentralized Probabilistic Auto-Scaling (DEPAS) algorithm can be used to deploy large-scale service systems whose scalability is limited only by the amount of virtualized resources that can be rented from IaaS providers. DEPAS assumes that the computing nodes are organized into an unstructured overlay network and run a scalable load balancing algorithm. DEPAS is run by each node which probabilistically decides to scale in and out. The main problem in DEPAS is how to compute this probability.

In \cite{calcavecchia2012} and \cite{caprarescu2012}, DEPAS was formulated, tested, and theoretically analyzed for the simplified case of homogenous systems. In this paper, we provided scaling probabilities formulas that work in the general case of heterogeneous systems. We proved both theoretically and experimentally that, by using the proposed formulas, DEPAS reacts to workload variations by allocating the right capacity in the first place.    
 
As future work, this paper will be extended with more and larger scale experimental results.

\section*{Acknowledgment}
This research has been partially funded by the Romanian National Authority for Scientific Research, CNCS Ð UEFISCDI, under project PN-II-ID-PCE-2011-3-0260 (AMICAS) and by the European Commission, under project FP7-ICT-2009-5-256910 (mOSAIC). Bogdan Caprarescu is partially supported by IBM through a PhD Fellowship Award.

We would like to thank to Nicola Calcavecchia and Daniel Dubois for their contribution to the development of the DEPAS simulator.

\bibliographystyle{IEEEtran}
\bibliography{IEEEabrv,adaptive2012}

\end{document}